\documentclass[sn-nature]{sn-jnl}

\usepackage{graphicx}%
\usepackage{multirow}%
\usepackage{amsmath,amssymb,amsfonts}%
\usepackage{amsthm}%
\usepackage{mathrsfs}%
\usepackage[title]{appendix}%
\usepackage{xcolor}%
\usepackage{textcomp}%
\usepackage{manyfoot}%
\usepackage{booktabs}%
\usepackage{algorithm}%
\usepackage{algorithmicx}%
\usepackage{algpseudocode}%
\usepackage{listings}%
\usepackage{paralist}
\usepackage{cleveref}
\usepackage{comment}

\theoremstyle{plain}%
\newtheorem{theorem}{Theorem}
\newtheorem{lemma}[theorem]{Lemma}
\newtheorem{question}{Question}

\theoremstyle{remark}%
\newtheorem{example}[theorem]{Example}%

\theoremstyle{definition}%

\newcommand{\maxcut}[0]{\textrm{Maximum Cut}}
\newcommand{\block}[0]{\mathcal{B}}
\newcommand{\intmodel}[1]{\mathcal{#1}}
\newcommand{\cut}[0]{\textrm{Cut}}
\newcommand{\zone}[1]{\textrm{Zone}(#1)}
\newcommand{\buffer}[2]{\textrm{Buffer}(#1,#2)}
\newcommand{\col}[0]{\chi}

\newcommand{\len}[1]{\textrm{len}({#1})}

\graphicspath{{figures/} }

\begin{document}

\title[Maximum Cut on Interval Graphs of Interval Count Two is NP-complete]{Maximum Cut on Interval Graphs of Interval Count Two is NP-complete}


\author*[1]{\fnm{Alexey} \sur{Barsukov}}\email{alexey.barsukov@matfyz.cuni.cz}

\author*[2]{\fnm{Bodhayan} \sur{Roy}}\email{bodhayan.roy@gmail.com}

\affil[1]{\orgdiv{Department of Algebra}, \orgname{Charles University, Czechia}}

\affil[2]{\orgdiv{Department of Mathematics}, \orgname{Indian Institute of Technology Kharagpur, India}}


\abstract{An interval graph has interval count $\ell$ if it has an interval model, where among every $\ell+1$ intervals there are two that have the same length. Maximum Cut on interval graphs has been found to be NP-complete recently by Adhikary et al.~\cite{interval} while deciding its complexity on unit interval graphs (graphs with interval count one) remains a longstanding open problem. More recently, de Figueiredo et al.~\cite{interval4} have made an advancement by showing that the problem remains NP-complete on interval graphs of interval count four. In this paper, we show that Maximum Cut is NP-complete even on interval graphs of interval count two.}

\keywords{maximum cut, interval graph, computational complexity}

\pacs[MSC Classification]{35A01, 65L10, 65L12, 65L20, 65L70}

\maketitle

\section{Introduction}

For a given input graph $G = \bigl(V(G), E(G)\bigr)$, the Maximum Cut problem finds a partition of $V(G)$ in two sets such that the number of edges between them is maximal.

$\maxcut$ is a fundamental and well-known NP-complete problem~\cite{garey_johnson1990}. 
The weighted version of the problem is one of Karp’s original 21 NP-complete problems~\cite{karp1972}. 
There are many results describing the complexity of $\maxcut$, where the input of the problem is restricted on some particular classes of graphs.
The problem remains NP-complete for many graph classes, such like cubic graphs~\cite{cubic}, split graphs~\cite{split_cobipartite_btw}, co-bipartite graphs~\cite{split_cobipartite_btw}, unit disk graphs~\cite{unit_disk}, total graphs~\cite{total_line}, and interval graphs~\cite{interval}. 
On the positive side, polynomial time algorithms are known for planar graphs~\cite{planar}, line graphs~\cite{total_line}, graphs not contractible to $K_5$~\cite{not_contractible_to_k5} and graphs with bounded treewidth~\cite{split_cobipartite_btw}.

There are two papers that have mainly motivated our research. 
First, a recent proof of NP-completeness of $\maxcut$ on interval graphs provided by Adhikary et al. in~\cite{interval}. 
Second, a more recent result of de Figueiredo et al. in~\cite{interval4}, where they extend the result of the first paper by proving that $\maxcut$ is NP-complete on graphs of interval count four. 
Using the technique of the above work, de Figueiredo et al. prove the NP-completeness of $\maxcut$ on permutation graphs as well, which too was open for a long time~\cite{permutation}.
The bounding of the number of interval lengths brings us closer to the final goal: to characterize $\maxcut$ for unit interval graphs as they are exactly interval graphs of interval count one. 
There were attempts to provide a polynomial-time algorithm for unit interval graphs~\cite{unit_interval1,unit_interval2}, but they both were later shown to be incorrect~\cite{unit_interval1correction,unit_interval2correction}. 
In this paper, we extend the result of de Figueiredo et al.~\cite{interval4} by proving Theorem~\ref{th:main_result} which brings us as close as possible to $\maxcut$ on unit interval graphs.
\begin{theorem}\label{th:main_result}
$\maxcut$ on interval graphs of interval count two is NP-complete.
\end{theorem}


\section{Preliminaries}\label{sec:preliminaries}

For $a,b$ in $\mathbb{R}$, an \emph{interval} between $a$ and $b$, denoted $[a,b]$, is the set of all $x\in\mathbb{R}$ such that $a\leq x\leq b$.
The \emph{left} and \emph{right endpoints} of $I = [a,b]$ are denoted by $\ell(I):=a$ and $r(I) := b$.
The \emph{length} $\len{I}$ of an interval $I$ equals $r(I)-\ell(I)$.
A set of intervals $\intmodel{I}$ is called an \emph{interval model} of a graph $G$ if there is a one-to-one mapping $f\colon V(G)\to\intmodel{I}$ such that, for any pair $u,v$ of vertices of $G$, $uv$ is in $E(G)$ if and only if the intervals $f(u)$ and $f(v)$ have a non-empty intersection.
A graph is called \emph{interval} if it has an interval model.
A graph $G$ has \emph{interval count} $\ell$ if there exists a set $C\subset \mathbb{R}$ of size $\ell$ and an interval model $\intmodel{I}$ of $G$ such that, for all $I$ in $\intmodel{I}$, we have $\len{I}\in C$.

A \emph{cut} $(R,B)$ of a graph $G$ is a partition of its vertices in two classes: $V(G) = R\sqcup B$.
For a cut $(R,B)$, a \emph{cut edge} is an edge that is incident to a vertex in $A$ and to a vertex in $B$.
The \emph{cut value} is the number of cut edges of the graph.
For a subset of vertices $S$ of a graph, the cut value \emph{contributed by} $S$ is the number of cut edges that have at least one of their ends in $S$.
Recall that the $\maxcut$ problem asks for a cut of the maximal value.
Notice that, for interval graphs, $\maxcut$ is equivalent to the problem of finding a coloring of an interval model in two colors, say $R$ (Red) and $B$ (Blue), where the number of differently colored pairs with non-empty intersection is maximal.
If an interval $I$ has a color $c\in\{R, B\}$, then we use the notation $\col(I) = c$. 
If $\block$ is a set of intervals that all have the same color $c\in\{R, B\}$, then we use the notation $\col(\block) = c$.

We use intervals of two different lengths: \emph{short} and \emph{long}. 
We assume that, if an interval $I$ is short, then $\len{I}=1$, otherwise, $\len{I}=\alpha$ for some $\alpha>1$ that will be made precise later.

A \emph{block} $\block$ is a set of short intervals such that, for any two intervals $I,I'\in\block$, we have $\ell(I)=\ell(I')$ and $r(I)=r(I')$.
For a block $\block$, denote $\ell(\block):=\ell(I)$ and $r(\block):=r(I)$ for an interval $I\in\block$.
The \emph{size} $|\block|$ of a block $\block$ is the number of intervals in it.
For example, a block of size $n$ is an interval model of the complete graph $K_n$ on $n$ vertices. 

A union of blocks $\block_1\cup\block_2\cup\block_3$ is a \emph{3-block} if
\begin{itemize}
    \item $r(\block_1)<\ell(\block_3)$ and $r(\block_1) \geq \ell(\block_2)$ and $r(\block_2) \geq \ell(\block_3)$ and
    \item $|\block_1|=|\block_3|$ and $|\block_2|=2|\block_1|$.
\end{itemize}
\begin{figure}[ht]
\centering
\includegraphics[width=0.5\textwidth]{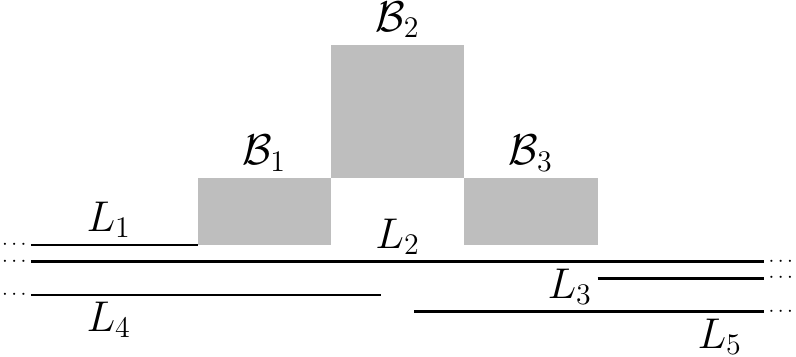}
\caption{A 3-block and five long intervals.}
\label{fig:3blocks0}
\end{figure}
In this article, 3-blocks are displayed as three rectangles, see \Cref{fig:3blocks0}. 
An interval $I$ \emph{terminates at} a block $\block$ if $\ell(\block)\leq r(I)\leq r(\block)$ and \emph{starts from} $\block$ if $\ell(\block)\leq \ell(I)\leq r(\block)$.
For a long interval $L$, we say that it \emph{covers} $\block$ if $\ell(L)<\ell(\block)<r(\block)<r(L)$.
For example, on \Cref{fig:3blocks0}, $L_1$ terminates at $\block_1$, $L_2$ covers all the blocks, $L_3$ starts from $\block_3$, $L_4$ covers $\block_1$ and terminates at $\block_2$, and $L_5$ starts from $\block_2$ and covers $\block_3$.

Consider any family of blocks of intervals of the same length $\block_1,\ldots,\block_s$ such that, for $i$ in $[s-1]$, we have $r(\block_i)=\ell(\block_{i+1})$.
A \emph{coloring} is a mapping $\chi\colon \block_1\cup\dots\cup\block_s \to \{R,B\}$.
A coloring $\chi$ is \emph{alternating} if, for $i\in[s-1]$, we have $\col(\block_i) \not= \col(\block_{i+1})$. 
A coloring $\chi$ is \emph{almost alternating except for $(\block_i,\delta)$} if it becomes alternating after removal of at most $\delta$ intervals from the block $\block_i$.

\section{Background}\label{section:background}

We first revisit the reduction of Adhikary et al. in~\cite{interval}. 
They reduced $\maxcut$ on cubic graphs to $\maxcut$ on interval graphs.
Recall that a graph $G$ is \emph{cubic} if the degree of every vertex $v\in V(G)$ is equal to 3.
In their paper, each vertex and edge of the original cubic graph was represented by a set of intervals, called \emph{vertex} and \emph{edge gadgets} respectively. 
The interval model consisted of first all the vertex gadgets, and then all the edge gadgets arranged from left to right. 
If an edge was incident to a vertex, then the corresponding vertex and edge gadgets were ``linked'' by a pair of very long intervals starting from the vertex gadgets and terminating at the edge gadget. 
They were called \emph{link intervals}. 
The number of intervals in any gadget was much greater than the total number of link intervals in the graph. 
It was shown that, in any $\maxcut$ partition of this interval graph, each vertex gadget or edge gadget could have only two possible partitions. 
For a vertex gadget, these two partitions were made to correspond to its membership in one of the partition sets for a maximum cut of the original cubic graph. 
If two adjacent vertices of the cubic graph belonged to different sets, then the corresponding edge gadget would make more cut edges with link intervals than if these vertex gadgets were in the same set.
Thus, a cut of maximum size of the given cubic graph always corresponded to a cut of maximum size of the constructed interval graph and vice versa.
As $\maxcut$ on cubic graphs is NP-complete, this reduction implies that it is also NP-complete on interval graphs.

In the reduction above, the gadgets contained intervals of two different lengths. 
However, the length of each link interval depended on the relative positions of its vertex and edge gadgets. 
So the total number of different lengths of link intervals was linearly dependent on the size of the cubic graph. 
So, this interval graph seemed to be far away from unit interval graphs, for which the problem was still open. 
De Figueiredo et al. in~\cite{interval4} made a very important advancement in this regard. 
They showed that $\maxcut$ was NP-complete even when the total number of lengths used for the intervals was only 4, i.e., when the interval count of the graph was 4.
In their paper, the vertex and edge gadgets were basically the same as that of Adhikary et al. in~\cite{interval}.
However, an extra gadget, resembling the vertex and edge gadgets, was used as a ``join gadget'' between link intervals. 
Instead of having a link interval running through the entire length between its corresponding vertex and edge gadgets, they used a chain of link intervals joined to each other with the use of join gadgets.
\begin{figure}[ht]
\centering
\includegraphics[width=0.99\textwidth]{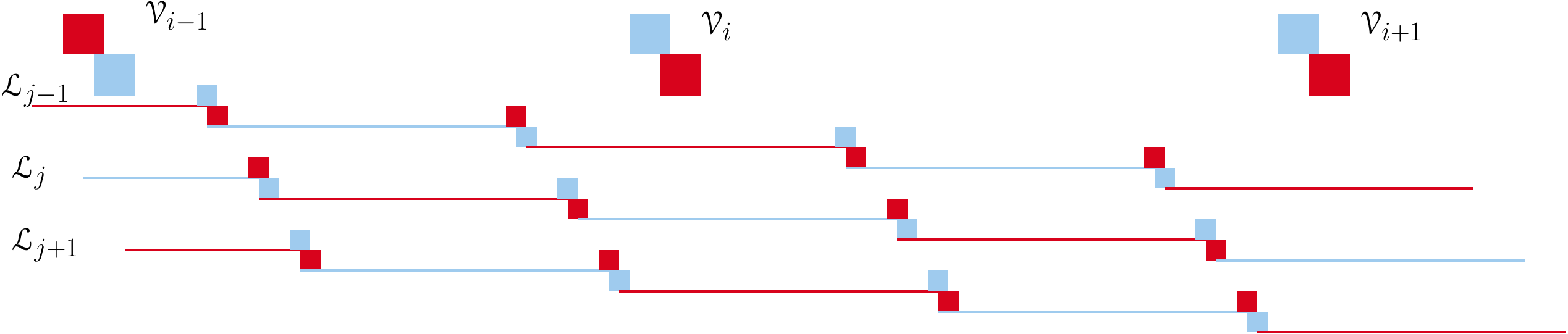}
\caption{Link chains between consecutive vertex gadgets. Squares represent families of intervals with the same endpoints.}
\label{fig:linkchains}
\end{figure}

A \emph{link chain} is a sequence of link intervals with a join gadget between every two consecutive link intervals, where one link interval terminates at the gadget and the other one starts from it. The link intervals of all other link chains that intersect this join gadget, cover it. The join gadgets ensure that, in a maximum cut partition, the link intervals of every link chain are colored with Red and Blue alternately. Thus, link intervals of arbitrary lengths in the original reduction can be replaced by link intervals of a single size. However, one problem remained. Between consecutive vertex or edge gadgets, the link chains had their join gadgets in a fixed order. For example, consider vertex gadgets $\intmodel{V}_{i-1}$, $\intmodel{V}_i$ and $\intmodel{V}_{i+1}$, and link chains $\intmodel{L}_{j-1}$, $\intmodel{L}_j$ and $\intmodel{L}_{j+1}$ (\cref{fig:linkchains}). If in the gap between $\intmodel{V}_{i-1}$ and $\intmodel{V}_i$, the join gadget of $\intmodel{L}_{i-1}$ occurred first, followed by a join gadget each of $\intmodel{L}_j$ and $\intmodel{L}_{j+1}$ respectively, then they would occur in the same order in the gap between $\intmodel{V}_i$ and $\intmodel{V}_{i+1}$. This would pose a problem if the structure of the graph required $\intmodel{L}_{j-1}$ and $\intmodel{L}_{j+1}$ to have a partial intersection with the same edge gadget. In such a case, a second type of link interval with a different length would be used to end link chain $\intmodel{L}_{j+1}$ early and enable it to partially intersect the same edge gadget along with $\intmodel{L}_{j-1}$. Thus, their reduction needed intervals of two lengths for the vertex, edge and join gadgets, and another two lengths as the link intervals, totaling to an interval count of four.


\section{Overview of the reduction}\label{section:reduction_overview}
We reduce the interval count down to 2 due to our modifications of gadgets and link chains. 
Instead of using separate lengths of short and long intervals in previous works, the gadgets are now mainly constructed from 3-blocks. 
The link chains do not use long intervals of two separate lengths anymore.
The issue of non-consecutive link chains partially intersecting the same edge gadget is resolved by using a \emph{switch gadget} that switches the relative positions of join gadgets of link chains. 

The $\maxcut$ problem is NP-complete on cubic graphs~\cite{cubic}.
In this article, we provide a polynomial time reduction from $\maxcut$ on cubic graphs to $\maxcut$ on interval graphs of interval count two.
For every cubic graph $G$ (of size $n$), we construct a graph $H$ of interval count 2 such that any $\maxcut$ partition of $H$ corresponds to some $\maxcut$ partition of $G$ and vice versa. 
Its top level composition is displayed on \Cref{fig:reduction}.
On the left, there are vertex gadgets $\intmodel{V}_1,\dots,\intmodel{V}_n$. 
For each vertex gadget there are three link chains constructed by alternating long intervals and join gadgets.
For every edge $g_ig_j$ of $E(G)$ there is an edge gadget $\intmodel{E}_{ij}$ such that two link chains $\intmodel{L}^i_{ij},\intmodel{L}_{ij}^j$ starting from $\intmodel{V}_i$ and $\intmodel{V}_j$ respectively both terminate at $\intmodel{E}_{ij}$.

Before $\intmodel{L}^i_{ij},\intmodel{L}_{ij}^j$ reach the edge gadget, we must ensure that there is no other link chain between them. 
The reason is that, in our construction, neither $\intmodel{E}_{ij}$ can intersect another gadget nor a long interval not from $\intmodel{L}^i_{ij}\cup\intmodel{L}_{ij}^j$ can start from or terminate at $\intmodel{E}_{ij}$.
Suppose, for example, that $v_1v_3\in E(G)$. 
If we straightforwardly connect link chains $\intmodel{L}^1_{13},\intmodel{L}_{13}^3$ to $\intmodel{E}_{13}$, then, as $\intmodel{V}_2$ is positioned ``between'' $\intmodel{V}_1$ and $\intmodel{V}_3$, this edge gadget will intersect join gadgets of link chains that start from $\intmodel{V}_2$.
To solve this issue, we repetitively apply the \emph{switching procedure} that switches two consecutive link chains and, therefore, reduces the number of join gadgets between the ones of $\intmodel{L}^i_{ij},\intmodel{L}_{ij}^j$. 
After these two chains terminate at $\intmodel{E}_{ij}$, we proceed similarly for every other edge of $E(G)$.  

\begin{figure}[ht]
    \centering
    \includegraphics[width=\textwidth]{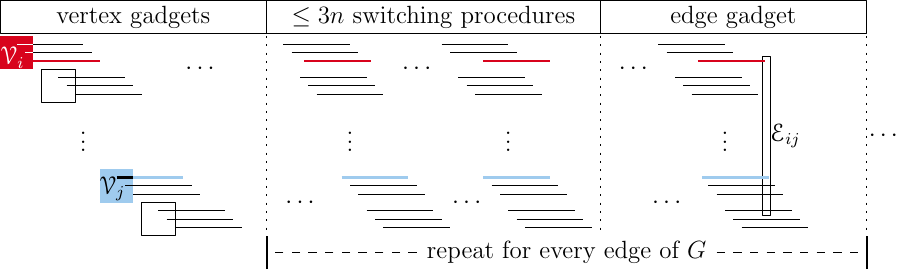}
    \caption{The composition of $H$.}
    \label{fig:reduction}
\end{figure}

\begin{example}
    Let $G$ be the cubic graph on 6 vertices displayed on Figure~\ref{fig:cubic graph}.
    
    \begin{figure}
        \centering
        \includegraphics[width=0.5\textwidth]{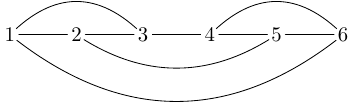}
        \caption{A cubic graph $G$.}
        \label{fig:cubic graph}
    \end{figure}

    The corresponding interval graph $H$ of interval count 2 is displayed on Figure~\ref{fig:construction_example}.
    On the left, it has 6 ``vertex gadgets'' $1,\dots,6$ that correspond to the vertices of $G$.
    For each vertex gadget, there are precisely three chains of long intervals that start from it, each chain represents an edge incident to the corresponding vertex.

    At first, we add the ``edge gadgets'' $E(12), E(23), E(34), E(45), E(56)$.
    The remaining edges are: $(13), (25), (46), (16)$.
    Notice that, for any of these edges, it is impossible to connect the corresponding long interval chains with an edge gadget without intersection of two gadgets.
    Therefore, we add two ``switch gadgets'' SWITCH(12) and SWITCH(45), and plug into them the corresponding long interval chains in order to switch their relative positions.
    After that, we add the edge gadgets $E(13)$ and $E(46)$.
    Finally, we deal with the edge $(25)$ and then with $(16)$.

    \begin{figure}
        \centering
        \includegraphics[scale=0.5,angle=90]{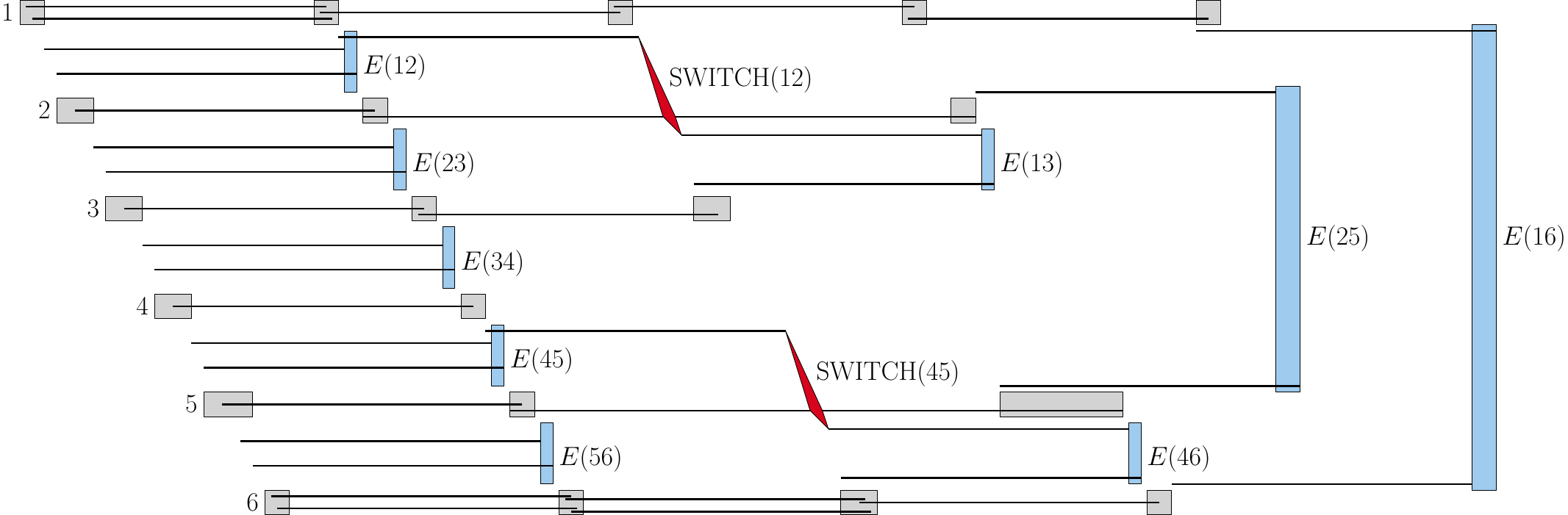}
        \caption{The graph $H$ of interval count two whose $\maxcut$ partitions correspond to $\maxcut$ partitions of the cubic graph $G$ from Figure~\ref{fig:cubic graph}.}
        \label{fig:construction_example}
    \end{figure}
\end{example}

\section{Gadgets}\label{section:gadgets}

Before we formally describe the reduction graph $H$ of interval count two, we present all the gadgets that are used in its construction.
Assuming that the size $|V(G)|$ of the cubic graph $G$ is $n$, we introduce two parameters that denote the gadget sizes.
Let $x$ be in $\Omega(n^7)$ and let $k$ be in $\Omega(n^6)$.
Next to the definition of each gadget, we add a figure, where this gadget is displayed.
On each of the figures, this gadget is colored in Red and Blue.
After we construct the whole graph $H$ of interval count two, we will argue that, for every $\maxcut$ partition of $H$, the coloring of each of its gadgets is similar to one displayed on the corresponding figure.

\paragraph*{Vertex gadget}
The  vertex gadget consists of three 3-blocks and two short intervals: $\intmodel{V}:=\block^1\sqcup\block^2\sqcup\block^3\sqcup \{S_{12},S_{23}\}$. 
The blocks $\block^i_1,\block^i_2,\block^i_3$ of each 3-block $\block^i$ have sizes $x,2x,x$ respectively. 
The 3-blocks are connected by short intervals $S_{12},S_{23}$, they are called \emph{short link} intervals.
That is, we have $\ell(S_{12})=r(\block^1_3)$ and $r(S_{12})=\ell(\block^2_1)$ and $\ell(S_{23}) = r(\block^2_3)$ and $r(S_{23}) = \ell(\block^3_1)$.
For every vertex gadget, there are three long intervals $L_1, L_2, L_3$ that start from $\block_3^1,\block_3^2,\block_3^3$ and form the corresponding link chains.
A vertex gadget and the three long intervals are displayed on \Cref{fig:vertex_gadget}.

\begin{figure}[ht]
    \centering
    \includegraphics[width=0.9\textwidth]{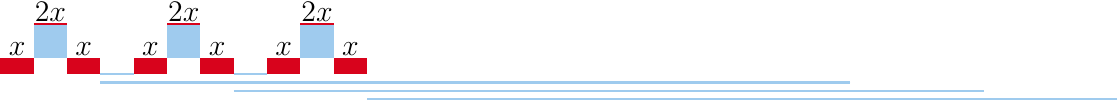}
    \caption{A vertex gadget.}
    \label{fig:vertex_gadget}
\end{figure}

\paragraph*{Edge gadget}
An edge gadget consists of two 3-blocks connected by a short link interval.
That is, $\intmodel{E}:=\block^1\sqcup\block^2\sqcup\{S_{12}\}$ such that $r(\block^1_3)=\ell(S_{12})$ and $r(S_{12}) = \ell(\block^2_1)$.
The blocks $\block^i_1,\block^i_2,\block^i_3$ of each 3-block $\block^i$ have sizes $k,2k,k$ respectively. 
There are two long intervals $L_1, L_2$ of the corresponding link chains terminating at $\intmodel{E}$ such that $r(L_1) = \ell(\block^1_1)$ and $r(L_2) = \ell(\block^2_2)$.
An edge gadget together with $L_1$ and $L_2$ is displayed on \Cref{fig:edge_gadget}.
Under a $\maxcut$ partition of $H$, for edge gadgets, the colors of $L_1$ and $L_2$ may be arbitrary.
So, \Cref{fig:edge_gadget} displays all the cases modulo switching the colors.

\begin{figure}[ht]
\centering
\includegraphics[width=0.65\textwidth]{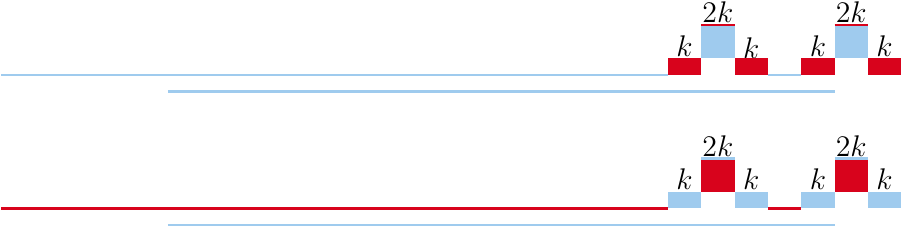}
\caption{The two cases of the $\maxcut$ coloring of an edge gadget: $\col(L_1)=\col(L_2)$ and $\col(L_1)\not=\col(L_2)$.}
\label{fig:edge_gadget}
\end{figure}

\paragraph*{Join gadget}

Similarly to vertex gadgets, a join gadget consists of 3-blocks of size $x,2x,x$ connected by short link intervals. 
It can have at most ten 3-blocks, see \Cref{fig:stretch_compress}. 
Every join gadget has at most 3 long intervals terminating at it and the same number of long intervals starting from it. 
The main purpose of a join gadget is to connect a vertex gadget to edge gadgets by link chains of long intervals. The distance between blocks may be modified in order to readjust the relative distances between long intervals terminating at the gadget, see \Cref{fig:stretch_compress}.

\begin{figure}[ht]
\centering
\includegraphics[width=\textwidth]{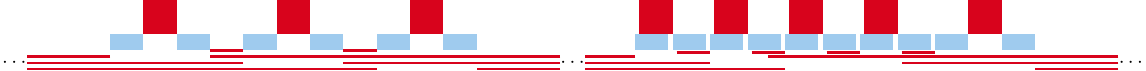}
\caption{A common join gadget and a join gadget readjusting the distances between long intervals.}
\label{fig:stretch_compress}
\end{figure}

\paragraph*{Switch gadget}
A switch gadget is the only gadget that is not constructed exclusively from 3-blocks. 
It has 2 parts. 
The \emph{first part} consists of 9 bottom blocks $\block^{\text{bot}}_1,\ldots,\block^{\text{bot}}_9$ and 4 top blocks $\block_1^{\text{top}},\ldots,\block_4^{\text{top}}$. 
At the bottom, $|\block_1^{\text{bot}}|=|\block_9^{\text{bot}}|=x$, and, for $2\leq i\leq 8$, we have $|\block_i^{\text{bot}}|=2x$. 
At the top, for $i\in\{1,\ldots,4\}$, we have $|\block_i^{\text{top}}|=2x'$, where $x'$ is such that $x/2<x'<x$. 
There are two long intervals $L_1$ and $L_2$ such that $r(L_1) = \ell(\block_1^{\text{bot}})$ and $r(L_2) = \ell(\block_1^{\text{top}})$.
There are two long intervals $R_1$ and $R_2$ such that $\ell(R_1) = r(\block_4^{\text{top}})$ and $\ell(R_2) = r(\block_9^{\text{bot}})$.
The long intervals $L_1, R_2$ are in the same link chain $\intmodel{L}$, the long intervals $L_2,R_1$ are in the same link chain $\intmodel{L}'$ distinct from $\intmodel{L}$.
The main property of the first part is that, for any $\maxcut$ partition, $\col(L_1) = \col(R_2)$ and $\col(L_2)\not=\col(R_1)$.
The \emph{second part} of the switch gadget is a 3-block $\block^2 = (\block^2_1,\block^2_2,\block^2_3)$ such that $R_1$ terminates at $\block^2_1$, $R_2$ covers the 3-block $\block^2$, and a new long interval $R_3$ starts from $\block^2_2$.
Both parts are displayed on \Cref{fig:switch_gadget}. 
The idea about this gadget is that it ``switches'' the link chains $\intmodel{L}$ containing $L_1R_2$ and $\intmodel{L}'$ containing $L_2R_1R_3$, and preserves their colors at the same time: $L_1$ is to the left from $L_2$ but $R_2$ is to the right from $R_1$, $\col(L_2)=\col(R_3)$ and $\col(L_1)=\col(R_2)$.

\begin{figure}[ht]
    \centering
    \includegraphics[scale=0.6]{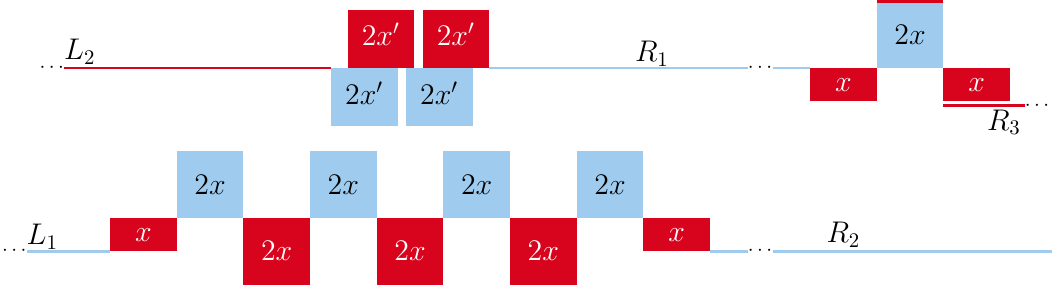}
    \caption{The first and the second parts of a switch gadget.}
    \label{fig:switch_gadget}
\end{figure}

The following lemma ensures that the coloring of the first part on \Cref{fig:switch_gadget} is $\maxcut$.

\begin{lemma}\label{lemma:switch_gadget}
Consider a graph with an interval model as on \Cref{fig:switch_gadget}. Then, in any $\maxcut$ partition of this graph, 
\begin{inparaenum}[(i)]
    \item for each of the two levels, block colorings are alternating, and 
    \item $\col(R_1)\not=\col(L_2)$, and $\col(R_2)=\col(L_1)$.
\end{inparaenum}
\end{lemma}

\begin{proof}
At first, let us compute the size of the cut for any partition, where the colors of the blocks on the top and on the bottom alternate. 
We do not consider the four long intervals at the moment. 
The value is the same for any of the four such partitions:
\begin{equation*}
    \cut_{\text{alter}} = 3\cdot 4x'^2 + 2\cdot 2x^2 + 6\cdot 4x^2 + 4\cdot 4x'x = 12x'^2 + 28x^2 + 16x'x.
\end{equation*}

Suppose now that $\cut_{\text{alter}}$ is not maximal. 
Then we will introduce the following parameters for each block.
\begin{compactitem}
\item Denote by $y_1,y_3$ the numbers of Blue intervals in $\block^{\text{top}}_1$ and $\block^{\text{top}}_3$.
\item Similarly, $y_2,y_4$ denote the numbers of Red intervals in $\block^{\text{top}}_2$ and $\block^{\text{top}}_4$.
\item Denote by $z_1,z_3,z_5,z_7,z_9$  the numbers of Blue intervals in $\block_1^{\text{bot}},\block_3^{\text{bot}},\block_5^{\text{bot}},\block_7^{\text{bot}},$ and in $\block_9^{\text{bot}}$.
\item  Similarly, denote by $z_2,z_4,z_6,z_8$ the numbers of Red intervals in $\block_2^{\text{bot}},\block_4^{\text{bot}},\block_6^{\text{bot}},$ and in $\block_8^{\text{bot}}$.
\end{compactitem}
We are going to compute the $\maxcut$ for the general case and to compare it to $\cut_{\text{alter}}$ in order to find out when it can be the maximal. 
We are going to split the computation into several parts.
\begin{compactitem}
    \item The $\cut_{\text{inside}}$ part counts the cut-edges inside each of the blocks.
    \item The $\cut_y$ and $\cut_z$ parts count the cut-edges between different blocks that are both on the same level. There are two levels: the top $y$ and the bottom $z$.
    \item The $\cut_{\text{inter}}$ part counts the cut-edges between the two levels.
\end{compactitem}

Now we compute the parts one by one.
\begin{multline*}
    \cut_{\text{inside}} = \sum_{i=1}^4 y_i(2x'-y_i) +z_1(x-z_1)+\sum_{j=2}^8 z_j(2x-z_j) +z_9(x-z_9)  \\
    = -\sum_{i=1}^4 y_i^2-\sum_{j=1}^9z_j^2+2x'\sum_{i=1}^4y_i+x(z_1+2\sum_{j=2}^8z_j+z_9).
\end{multline*}
\begin{multline*}
    \cut_y = \sum_{i=1}^3\bigl(y_iy_{i+1} + (2x'-y_i)(2x'-y_{i+1})\bigr) \\
    =2\sum_{i=1}^3y_iy_{i+1} + 12x'^2 - x'(2y_1+4\sum_{i=2}^3y_i + 2y_4).
\end{multline*}
\begin{multline*}
    \cut_z = z_1z_2+(x-z_1)(2x-z_2)+\sum_{j=2}^7\bigl(z_jz_{j+1}+(2x-z_j)(2x-z_{j+1})\bigr)\\ 
    +z_8z_9+(2x-z_8)(x-z_9)= 2\sum_{j=1}^8z_jz_{j+1}+28x^2-x(2z_1+3z_2+4\sum_{j=3}^7z_j+3z_8+2z_9).
\end{multline*}
Denote by $A$ the set of pairs $\bigl\{(i,j)\mid \block^{\text{top}}_i\text{ intersects }\block^{\text{bot}}_j\bigr\}$. That is,
\begin{equation*}
A = \Bigl\{(1,4),(1,5),(2,4),(2,5),(3,5),(3,6),(4,5),(4,6)\Bigr\}.
\end{equation*}
Then
\begin{multline*}
    \cut_{\text{inter}} = \sum_{(i,j)\in A,i+j\text{ is odd}}\bigl(y_iz_j +(2x'-y_i)(2x-z_j)\bigr) \\
    +\sum_{(i,j)\in A,i+j\text{ is even}}\bigl(y_i(2x-z_j)+(2x'-y_i)z_j\bigr) =    16x'x + 2\sum_{i,j\in A}(-1)^{i+j+1}y_iz_j \\
    =16x'x + 2(y_1-y_2)(z_4-z_5) -2(y_3-y_4)(z_5-z_6).
\end{multline*}
Our goal is to prove that $f=\cut_{\text{inside}}+\cut_y+\cut_z+\cut_{\text{inter}}-\cut_{\text{alter}}$ is less or equal to 0 and the equality is reached only when the colors alternate. 
Think of $f$ as of a polynomial in $x'$ and $x$ of degree 2:
\begin{equation*}
    f = f(x',x) = f_0 + f_1 x' + f_2 x + f_3 x'^2 + f_4 x'x + f_5 x^2.
\end{equation*}
Clearly, $f_3=f_4=f_5=0$. Compute the terms $f_0$, $f_1 x'$, and $f_2 x$:
\begin{multline*}
    f_0 = \underbrace{-\sum_{i=1}^4 y_i^2-\sum_{j=1}^9z_j^2}_{\cut_{\text{inside}}} + \underbrace{2\sum_{i=1}^3y_iy_{i+1}}_{\cut_y} + \underbrace{2\sum_{j=1}^8z_jz_{j+1}}_{\cut_z} \\
    +\underbrace{2(y_1-y_2)(z_4-z_5) -2(y_3-y_4)(z_5-z_6)}_{\cut_{\text{inter}}};
\end{multline*}
\begin{equation*}
    f_1 x' = \underbrace{2x'\sum_{i=1}^4y_i}_{\cut_{\text{inside}}}- \underbrace{x'(2y_1+4\sum_{i=2}^3y_i + 2y_4)}_{\cut_y} =-2x'\sum_{i=2}^3y_i= -\sum_{i=2}^3 y_i^2 - \sum_{i=2}^3(2x'-y_i)y_i;
\end{equation*}
\begin{multline*}
    f_2 x = \underbrace{x(z_1+2\sum_{j=2}^8z_j+z_9)}_{\cut_{\text{inside}}} - \underbrace{x(2z_1+3z_2+4\sum_{j=3}^7z_j+3z_8+2z_9)}_{\cut_z} \\
    = -x(z_1+z_2+2\sum_{j=3}^7z_j+z_8+z_9) \\ =-z_1^2-\frac{z_2^2}{2}-\sum_{j=3}^7z_j^2-\frac{z_8^2}{2}-z_9^2 \\ -(x-z_1)z_1-\left(x-\frac{z_2}{2}\right)z_2-\sum_{j=3}^7(2x-z_j)z_j - \left(x-\frac{z_8}{2}\right)z_8-(x-z_9)z_9.
\end{multline*}

We have  extracted the negative squares of $y_i,z_j$ from $f_1 x'$ and $f_2 x$ in order to combine them together with the negative squares of the part of $f_0$ provided by $\cut_{\text{inside}}$. 
The other summands of $f_1 x'$ and $f_2 x$ are almost always negative, except for the minimal and maximal values of $y_i,z_j$. 
These cases happen exactly when a block is either all Red or all Blue. 
Then
\begin{multline*}
    f = f_0 + f_1 x' + f_2 x= 2\sum_{i=1}^3y_iy_{i+1} - y_1^2-2\sum_{i=2}^3y_i^2-y_4^2\\
    +2\sum_{j=1}^8z_jz_{j+1} - 2z_1^2-\frac{3}{2}z_2^2-2\sum_{j=3}^7z_j^2-\frac{3}{2}z_8^2-2z_9^2\\
    + 2(y_1-y_2)(z_4-z_5) -2(y_3-y_4)(z_5-z_6)\\
    -\sum_{i=2}^3(2x'-y_i)y_i\\
    -(x-z_1)z_1-\left(x-\frac{z_2}{2}\right)z_2-\sum_{j=3}^7(2x-z_j)z_j - \left(x-\frac{z_8}{2}\right)z_8-(x-z_9)z_9
\end{multline*}
\begin{multline*}
    =-\sum_{i=1}^3(y_i-y_{i+1})^2 - 2\left(z_1-\frac{z_2}{2}\right)^2-\sum_{j=2}^7(z_j-z_{j+1})^2 - 2\left(\frac{z_8}{2}-z_9\right)^2\\
    + 2(y_1-y_2)(z_4-z_5) -2(y_3-y_4)(z_5-z_6)\\
    -\sum_{i=2}^3(2x'-y_i)y_i\\
    -(x-z_1)z_1-\left(x-\frac{z_2}{2}\right)z_2-\sum_{j=3}^7(2x-z_j)z_j - \left(x-\frac{z_8}{2}\right)z_8-(x-z_9)z_9
\end{multline*}
\begin{multline*}
    = -(y_1-y_2-z_4+z_5)^2-(y_2-y_3)^2 -(y_3-y_4+z_5-z_6)^2\\
    - 2\left(z_1-\frac{z_2}{2}\right)^2 -(z_2-z_3)^2 -(z_3-z_4)^2 -(z_6-z_7)^2-(z_7-z_8)^2- 2\left(\frac{z_8}{2}-z_9\right)^2\\
    -\sum_{i=2}^3(2x'-y_i)y_i\\
    -(x-z_1)z_1-\left(x-\frac{z_2}{2}\right)z_2-\sum_{j=3}^7(2x-z_j)z_j - \left(x-\frac{z_8}{2}\right)z_8-(x-z_9)z_9.
\end{multline*}
Clearly, $f\leq0$. We are going to find all the cases when $f=0$. 
Notice that, for any $i\in\{2,3\},j\in\{1,\ldots,9\}$, there is a summand $-y_i(2x'-y_i)$ and a corresponding one for $z_j$. 
Thus, we need to consider only the cases when $y_2,y_3\in\{0,2x'\}$, $z_1,z_9\in\{0,x\}$ and $z_2,\ldots,z_8\in\{0,2x\}$. 
Suppose that $y_1\not\in\{0,2x'\}$; then $0<|y_1-y_2|<2x'$ and thus $y_1-y_2-z_2+z_3\not=0$, for any $z_2,z_3\in\{0,2x\}$, as $x'<x$. 
So $y_1=y_2$, and thus $z_2=z_3$. 
Similarly, $y_3=y_4$, and so $z_5=z_6$. 
Using other clauses of the expression, we conclude that $y_1=\ldots=y_4$, $z_2=\ldots=z_8$, and also we have $z_1=z_9=z_2/2=z_8/2$. 
This means that the colorings of the blocks alternate.

Now we need to prove that $\col(L_2)\not=\col(R_1)$. 
For the contrary, suppose that $\col(L_2)=\col(R_1)$. 
Then the cut between the blocks and the long intervals will be at most the sum: 
\begin{equation*}
4x + 4x + x + x + 2x' = 10x + 2x',\text{ where}
\end{equation*}
$4x$ is provided by $L_2$, $4x$ is provided by $R_1$,  $x$ is provided by $L_1$, $x$ is provided by $R_2$, and  $2x'$ provided by one of $L_2,R_1$. 
On the other hand, if $\col(L_2)\not=\col(R_1)$, then together $L_2$ and $R_1$ provide $7x+4x'$, for each of two possible colorings of the bottom blocks, and $L_1$ and $R_2$ provide $2x$, in total it is $9x+4x'$. 
This case is reachable, because we can choose the leftmost bottom block to have a color opposite to $\col(L_1)$. As $\col(R_2)$ is not fixed, we can choose the right one that will add $x$ as well. 
So we have to satisfy the inequality
\begin{equation*}
    10x+2x' < 9x + 4x'.
\end{equation*}
We assumed that $2x'>x$ so the inequality holds and thus the second case is optimal. 
We should note that, in the second case, $L_1$ and $R_2$ have the same color because $\col(\block^{\bot}_1)=\col(\block^{\bot}_9)$. 
We have shown that $L_1$ and $R_2$ are colored similarly and that $L_2$ and $R_1$ are colored differently. 
\end{proof}

\section{Construction of the graph}




Let $G$ be a cubic graph of size $n$. Let the long interval length be $\alpha := 53n-3$ while the short interval length is 1. 
All the intervals of $H$ belong to $[0,+\infty)$. 
This ray is split into zones and buffers in the following order:
\[
\zone{1},\buffer{1}{2},\zone{2},\buffer{2}{3},\ldots,\zone{n},\buffer{n}{1},\zone{1},\ldots
\]

Each zone (or buffer) is a disjoint union of \emph{fragments} of the same length, where the distance between two neighbor fragments of the same zone is the same and depends on the long interval length $\alpha$. 

\begin{compactitem}
    \item Here, $\zone{1},\ldots,\zone{n}$ are the zones that correspond to vertices of $G$. For a vertex $g_i\in V(G)$, $\zone{i}$ is the following disjoint union of fragments of the same size:
    \begin{equation*}
        \zone{i} = \bigcup_{j\in\mathbb{Z}} \bigl[53i + j(\alpha+3), 53i+j(\alpha+3)+32\bigr].
    \end{equation*}
    This zone usually contains the vertex and the join gadgets corresponding to $g_i\in V(G)$. The size of its fragments is 32. The distance $\alpha+3$ between the left endpoints of two neighbor fragments is called the \emph{phase}.
    
    \item Buffer zones. For two vertices $g_i,g_{i+1}\in V(G)$ (and also $g_{n-1},g_0$), we introduce a buffer zone between $\zone{i},\zone{i+1}$. It is denoted by $\buffer{i}{i+1}$ ($\buffer{n-1}{0}$ is denoted similarly) and is defined as follows:
    \begin{equation*}
        \buffer{i}{i+1} = \bigcup_{j\in\mathbb{Z}}\bigl[53i+j(\alpha+3)+32, 53(i+1)+j(\alpha+3)\bigr].
    \end{equation*}
    We need buffer zones in order to do the switching. Every fragment of a buffer zone has size 21.
\end{compactitem}
We need to write $\alpha+3$ instead of $\alpha$ because a long interval must start from the rightmost block of some 3-block and must terminate at the leftmost block of another 3-block. So, the length of long intervals must be lesser than the length of the phase by 3.

For $i\in\{1,\ldots,n\}$, add a vertex gadget $\intmodel{V}_i$ into the leftmost fragment of $\zone{i}$. For each $\intmodel{V}_i$ there are 3 long intervals starting from it; they terminate at a join gadget that has 3 new long intervals starting from it, and so on. This produces 3 link chains of long intervals. Each such chain eventually terminates at a corresponding edge gadget.

For every edge $g_ig_j\in E(G)$, where $i<j$, we choose a link chain starting from $\intmodel{V}_i$ and repeatedly apply  the switch procedure (introduced formally in the next paragraph) to this chain until it is in $\zone{j}$.  
Once it happens, this chain of $\intmodel{V}_i$ and one of the chains of $\intmodel{V}_j$ terminate at the same edge gadget $\intmodel{E}_{ij}$.  
No long interval starts from $\intmodel{E}_{ij}$. 
Then we choose another edge of $G$ and repeat this operation until all edges of $G$ are treated. 
If a chain does not participate in some switch procedure, then, during this procedure, the long intervals of this chain are joined by join gadgets that look the same as vertex gadgets.  
The composition of $H$ is displayed on \Cref{fig:reduction} on \cpageref{fig:reduction}.

\paragraph*{Switch Procedure}

Consider some adjacent vertices $g_i,g_j$ of the cubic graph $G$, they correspond to two vertex gadgets $\intmodel{V}_i,\intmodel{V}_j$ and to an edge gadget $\intmodel{E}_{ij}$ in the interval model of $H$. 
Suppose that there is another vertex gadget $\intmodel{V}'$ between $\intmodel{V}_i$ and $\intmodel{V}_j$. 
Then, there is a join gadget that corresponds to $\intmodel{V}'$ between any pair of join gadgets of $\intmodel{V}_i,\intmodel{V}_j$. 
We require that a gadget can be intersected only by long intervals, in particular, we forbid any gadget to intersect $\intmodel{E}_{ij}$, so we have to switch some link chain starting from $\intmodel{V}_i$ with all three link chains that start from $\intmodel{V}'$, and similarly for any other vertex gadget between $\intmodel{V}_i$ and $\intmodel{V}_j$. 

We solve it by the \emph{switch procedure} which is displayed on \Cref{fig:switch_procedure}. 
The line $\mathbb{R}$ is split into zones and buffers between zones. 

For $i\in\{1,\ldots,n\}$, $\zone{i}$ corresponds to the vertex gadget $\intmodel{V}_i$, it contains $\intmodel{V}_i$ and most of the join gadgets of the link chains starting from $\intmodel{V}_i$. 
Every fragment of $\buffer{i}{i+1}$ is placed between the fragments of $\zone{i}$ and $\zone{i+1}$. 
$\buffer{i}{i+1}$ contains switch gadgets that are used to let a chain starting from $\intmodel{V}_i$ pass all three chains starting from $\intmodel{V}_{i+1}$, this is exactly what is shown on \Cref{fig:switch_procedure}.
By repeatedly iterating the switching, we pass all the vertex gadgets between $\intmodel{V}_i$ and $\intmodel{V}_j$ and eventually connect the corresponding link chains to a common edge gadget $\intmodel{E}_{ij}$. 
In order to justify the correctness of \Cref{fig:switch_procedure}, we describe the precise positions of each gadget and long interval below.

\begin{figure}[ht]
    \centering
    \includegraphics[width=\textwidth]{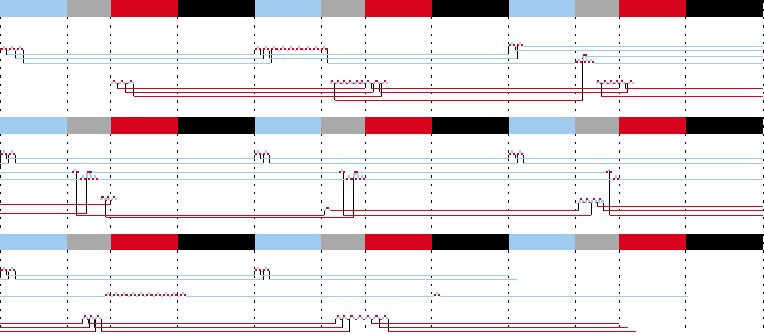}
    \caption{The switch procedure. $\zone{i}$ is Blue. $\zone{i+1}$ is Red. $\buffer{i}{i+1}$ between them is gray. Black fragment contains all other zones and buffers. {\bf One should read this figure from left to right and from top to bottom, just like a book.}}
    \label{fig:switch_procedure}
\end{figure}


Consider two vertices $g_i,g_{i+1}\in V(G)$. $\zone{i}$ is colored in Blue, $\buffer{i}{i+1}$ is colored in gray, $\zone{i+1}$ -- in Red, the black color corresponds to the rest -- the union of all other zones and buffers. Suppose that $g_i$ is adjacent to some $g_j$, for $j\not\in\{i-1,i+1\}$. Then, for some chain of long intervals that starts from the vertex gadget $\intmodel{V}_i$ of $g_i$, we need to put it to the right of $\zone{i+1}$, that is to switch it with the long interval chains starting from $\intmodel{V}_{i+1}$. On \Cref{fig:switch_procedure}, it is equivalent to move one Blue interval to the black segment by passing the orange segment beforehand. We suppose that the long intervals starting from $\intmodel{V}_i$ are Blue and those starting from $\intmodel{V}_{i+1}$ are Red, and so are all other long intervals that are connected to these gadgets by a chain of join gadgets. 
Say that a gadget is in the fragment $[a,b]$ of some zone if its leftmost and rightmost points are in $a$ and $b$. We are going to show that no gadgets intersect each other by considering each of nine buffer zones on \Cref{fig:switch_procedure}.

\begin{enumerate}
    \item The first buffer does not contain any gadgets.
    \item The second buffer intersects two join gadgets. The Blue one is in $[0,32]$ of $\zone{i}$ and in $[0,3]$ of $\buffer{i}{i+1}$. Its long intervals terminate at $\{0,4,8\}$ of $\zone{i}$ and start from $\{3,7\}$ of $\zone{i}$ and from $3$ of $\buffer{i}{i+1}$. The Red one is in $[4,21]$ of $\buffer{i}{i+1}$ and in $[0,11]$ of $\zone{i+1}$. Its long intervals terminate at $\{0,4,8\}$ of $\zone{i+1}$ and start from $6.5$ of $\buffer{i}{i+1}$ and from $\{3,7\}$ of $\zone{i+1}$.
    \item The third buffer contains the first switch gadget in $[0,9]$. Its long intervals terminate at $\{0,3.5\}$ and start from $\{5.5,9\}$ of $\buffer{i}{i+1}$. It also contains the Red join gadget in $[10,21]$ of $\buffer{i}{i+1}$ and in $[0,7]$ of $\zone{i+1}$. Its long intervals terminate at $\{0,4\}$ of $\zone{i+1}$ and start from $12.5$ of $\buffer{i}{i+1}$ and from $3$ of $\zone{i+1}$.
    \item The fourth buffer contains  the second switch gadget in $[6,15]$. Its long intervals terminate at $\{6,9.5\}$ and start from $\{11.5,15\}$. It also contains the second part of the first switch gadget in $[2.5,5.5]$. Its long intervals terminate at $2.5$ and start from $4.5$. It also contains a Red join gadget in $[16,21]$ of $\buffer{i}{i+1}$ and in $[0,3]$ of $\zone{i+1}$. Its long intervals terminate at $21$ and start from $18.5$ of $\buffer{i}{i+1}$.
    \item The fifth buffer contains the second part of the second switch gadget in $[8.5,11.5]$. Its long intervals terminate at $8.5$ and start from $10.5$. It also contains the third switch gadget in $[12,21]$. Its long intervals terminate at $\{12,15.5\}$ and start from $\{17.5,21\}$. It also contains a Red join gadget in $[1.5,4.5]$. Its long intervals terminate at $1.5$ and start from $4.5$.
    \item The sixth buffer contains a Blue join gadget in $[18,21]$. Its long intervals terminate at $18$ and start from $21$. It also contains the second part of the third switch gadget in $[14.5,17.5]$. Its long intervals terminate at $14.5$ and start from $16.5$. It also contains a Red join gadget in $[1.5,13.5]$. Its long intervals terminate at $1.5,7.5$ and start from $10.5,13.5$.
    \item The seventh buffer intersects a Blue join gadget that is in $[18,21]$ of $\buffer{i}{i+1}$, in $[0,32]$ of $\zone{i+1}$ and in $[0,4]$ of $\buffer{i+1}{i+2}$. Its long intervals terminate at $18$ of $\buffer{i}{i+1}$ and start from $4$ of $\buffer{i+1}{i+2}$. It also contains a Red join gadget in $[7.5, 16.5]$. Its long intervals terminate at $\{7.5,10.5,13.5\}$ and start from $\{10,13,16.5\}$.
    \item The eighth buffer intersects a Red join gadget. It is in $[7,21]$ of $\buffer{i}{i+1}$ and in $[0,11]$ of $\zone{i+1}$. Its long intervals terminate at $\{7, 10, 13.5\}$ of $\buffer{i}{i+1}$ and start from $\{3,7,11\}$ of $\zone{i+1}$.
    \item The ninth buffer is empty.
\end{enumerate}

\section{Gadget partitions}

In this section, we describe and prove how exactly each gadget of $H$ can be colored under a $\maxcut$ partition of $H$.
Sometimes, for a property to hold, it is required that $H$ satisfies certain numerical constraints. 
We now list these constraints, and, later, when we provide an explicit construction of $H$, we make sure that all of them are satisfied. 
Suppose that the cubic graph $G$ has $n$ vertices. Then, assume the following.
\begin{enumerate}
\item Every block of $H$ is covered by at most $\eta := 3n$ long intervals. 
\item Every long interval of $H$ covers at most $\mu:=30n+16$ blocks.
\end{enumerate}

For each gadget of $H$, we assume that it is covered by at most $\eta$ long intervals.

\begin{lemma} \label{lemma:3block_stable}
Let $\intmodel{B} = (\block_1, \block_2, \block_3)$ be a 3-block in $H$ of size $(x,2x,x)$. Then, for any $\maxcut$ partition of $H$, $\col(\block_1)=\col(\block_3)$, and all except for at most $\eta$ intervals of $\block_2$ have the opposite color. 
\end{lemma}
\begin{proof}
Within the graph $H$, there are five ways how a long interval can intersect a 3-block.
All these five ways are displayed on \Cref{fig:3blocks0} on \cpageref{fig:3blocks0}.
Let $\intmodel{L}_1,\ldots,\intmodel{L}_5$ be the corresponding sets of long intervals that intersect a given 3-block within the graph $H$. That is,
\begin{compactitem}
\item $\intmodel{L}_1$ denotes all the long intervals that terminate at $\block_1$;
\item $\intmodel{L}_2$ denotes all the long intervals that cover $\block_1$, $\block_2$, and $\block_3$; 
\item $\intmodel{L}_3$ denotes all the long intervals that start from $\block_3$;
\item $\intmodel{L}_4$ denotes all the long intervals that terminate at $\block_2$;
\item $\intmodel{L}_5$ denotes all the long intervals that start from $\block_2$. 
\end{compactitem}

We consider some $\maxcut$ partition of $H$.
For $i$ in $[5]$, let $r_i$ and $b_i$ be the numbers of Red and Blue intervals in $\intmodel{L}_i$.
\emph{W.l.o.g.}, assume that most of the intervals of $\intmodel{L}_1$ get the color Red, that is, $r_1\geq b_1$.
Also let  $y_1,2x-y_2$, and $y_3$ be the number of intervals of respectively $\block_1,\block_2,\block_3$ that are colored Blue. Hence  $x-y_1,y_2$, and $x-y_3$ are  the numbers of intervals of respectively $\block_1,\block_2$, and $\block_3$ that are colored Red.
We want to show that one of the two following conditions must hold:
\begin{itemize}
    \item either $y_1=y_3=0$ and $y_2\in[0,\eta]$, or
    \item $y_1=y_3=x$ and $y_2\in[2x-\eta,2x]$.
\end{itemize}

Let us first calculate the cut value contributed by the 3-block.
It consists of the three following parts.
\begin{enumerate}

 \item  The number of cut edges formed between the blocks and the long intervals is $r_1y_1 + b_1(x-y_1) + r_2(y_1 + 2x - y_2 + y_3) + b_2(2x-y_1 + y_2 - y_3) + r_3y_3 + b_3(x-y_3) + r_4(y_1+2x-y_2) + b_4(x-y_1+y_2) + r_5(2x-y_2+y_3) + b_5(y_2+x-y_3)$.

 \item  The number of cut edges formed by the short intervals of the same block among themselves is $y_1(x-y_1) + y_2(2x-y_2) + y_3(x-y_3)$.

 \item  The number of cut edges formed by the short intervals of the different blocks is $y_1y_2 + (x-y_1)(2x - y_2) + y_2y_3 + (2x-y_2)(x-y_3)$.

\end{enumerate}

Denote the value $r_i-b_i$ by $\Delta_i$.
Denote by $C$ the cut value contributed by the 3-block $\intmodel{B}$, which is the sum of the three aforementioned numbers:
\begin{align*}
C = & 4x^2 + x(2r_2 + 2r_4 +2r_5 + b_1+2b_2+b_3+b_4+b_5) \\
& - xy_1 -xy_3-y_1^2 - y_2^2 - y_3^2+ 2y_1y_2 + 2y_2y_3 \\
&+ \Delta_1y_1 + \Delta_2(y_1-y_2+y_3) + \Delta_3y_3 + \Delta_4(y_1 - y_2) +\Delta_5(y_3-y_2).
\end{align*} 

Now let us modify the cut by making $\col(\block_1)=\col(\block_3)=R$ and $\col(\block_2)=B$. This means that we have $y_1 = y_2 = y_3 = 0$. So the cut value, denoted by $C_{RBR}$, contributed by the 3-block $\intmodel{B}$ in this case is:
\begin{align*}
C_{RBR} = 4x^2  + x(2r_2 + 2r_4 +2r_5 + b_1+2b_2+b_3+b_4 + b_5).
\end{align*}

If we instead have $\col(\block_1)=\col(\block_3)=B$ and $\col(\block_2)=R$ (i.e., $y_1 = y_3 = x, y_2 = 2x$), then the cut value, denoted by $C_{BRB}$, contributed by the 3-block $\intmodel{B}$ would be:
\begin{align*}
C_{BRB} = 4x^2  + x(r_1 + 2r_2 + r_3 + r_4 + r_5 + 2b_2 + 2b_4 + 2b_5).
\end{align*}

\emph{W.l.o.g.}, assume that $C_{RBR} \geq C_{BRB}$. 
This implies that $C_{RBR} - C_{BRB} \geq 0$, which means that $\Delta_1 +\Delta_3 - \Delta_4 - \Delta_5 \leq 0$.

Notice that $C - C_{RBR}$ equals
\begin{align*}
& - y_1^2 - y_2^2 - y_3^2 - xy_1 -xy_3 + 2y_1y_2 + 2y_2y_3\\
 &\hspace{12pt}+ \Delta_1y_1 + \Delta_2(y_1-y_2+y_3) + \Delta_3y_3 + \Delta_4(y_1 - y_2) + \Delta_5(y_3-y_2)\\
&= -\left[(y_1 - y_2 + y_3) - \frac{1}{2}(\Delta_2 + \Delta_4 + \Delta_5)\right]^2 + \frac{1}{4}(\Delta_2 + \Delta_4 + \Delta_5)^2  \\
& +2y_1y_3 - xy_1 -xy_3 + \Delta_1y_1  + \Delta_3y_3  - \Delta_4y_3 - \Delta_5y_1\\
&\leq \frac{1}{4}(\Delta_2 + \Delta_4 + \Delta_5)^2 + 2y_1y_3 - xy_1 -xy_3 + (\Delta_1-\Delta_5)y_1  + (\Delta_3-\Delta_4)y_3\\
&= \frac{1}{4}(\Delta_2 + \Delta_4 + \Delta_5)^2 + 2y_1y_3 - xy_1 -xy_3 + (\Delta_1-\Delta_5)(y_1-y_3) + (\Delta_1  + \Delta_3  - \Delta_4 - \Delta_5)y_3\\
&\leq \frac{1}{4}(\Delta_2 + \Delta_4 + \Delta_5)^2 + 2y_1y_3 - xy_1 -xy_3 + (\Delta_1-\Delta_5)(y_1-y_3)\\
&\hspace{12pt}(\text{since } \Delta_1  + \Delta_3  - \Delta_4 -\Delta_5 \leq 0)\\
&= \frac{1}{4}(\Delta_2 + \Delta_4 + \Delta_5)^2 + \frac{1}{4}(\Delta_1-\Delta_5)^2 + y_1(y_1-x) + y_3(y_3 - x) - \left(y_1 - y_3 -\frac{1}{2}(\Delta_1-\Delta_5)\right)^2  \\
&\leq \frac{1}{4}(\Delta_2 + \Delta_4 + \Delta_5)^2 + \frac{1}{4}(\Delta_1-\Delta_5)^2 + y_1(y_1-x) + y_3(y_3 - x).
\end{align*}

Notice that if $y_1 \notin \{0, x\}$, then $(x-y_1)y_1 \geq x -1$. Similarly, if $y_3 \notin \{0, x\}$, then $(x-y_3)y_3 \geq x -1$. Since we assumed that either $(y_1 \notin \{0, x\})$ or $(y_3 \notin \{0, x\})$, it implies that $y_1(y_1-x) + y_3(y_3 - x) \leq 1 - x $. So we have $C - C_{RBR} \leq \frac{1}{4}(\Delta_2 + \Delta_4 + \Delta_5)^2 + \frac{1}{4}(\Delta_1-\Delta_5)^2 + 1 - x$. But since $x > n^6$ and each $\Delta_i < n^3$, we conclude that $C - C_{RBR} < 0 \implies C < C_{RBR}$. This contradicts the maximality of $C$. Therefore, we conclude that $\bigl(y_1 \in \{0, x\}\bigr)$ and $\bigl(y_3 \in \{0, x\}\bigr)$.

Now we show that $y_1=y_3$. 
For the contrary, suppose that $y_1 \not= y_3$. 
Then, $|y_1 - y_3| = x$. 
In this case, the value $C - C_{RBR}$ is :
\begin{align*}
 & \frac{1}{4}(\Delta_2 + \Delta_4 + \Delta_5)^2 + \frac{1}{4}(\Delta_1-\Delta_5)^2 + y_1(y_1-x) + y_3(y_3 - x) - \left(y_1 - y_3 -\frac{1}{2}(\Delta_1-\Delta_5)\right)^2\\
&= \frac{1}{4}(\Delta_2 + \Delta_4 + \Delta_5)^2 + \frac{1}{4}(\Delta_1-\Delta_5)^2  - \left(y_1 - y_3 -\frac{1}{2}(\Delta_1-\Delta_5)\right)^2\\
&\leq \frac{1}{4}(\Delta_2 + \Delta_4 + \Delta_5)^2   - x^2 + x|\Delta_1-\Delta_5|.
\end{align*} 

Since $x > n^6$ and each $\Delta_i < n^3$, we have that $C - C_{RBR} < 0$, which  contradicts the first cut being maximal. So we conclude that $y_1=y_3$.

\emph{W.l.o.g.}, assume that $y_1 = y_3=0$. 
We want to show that $y_2\in[0,\eta]$.
As $y_1=y_3=0$, the cut value contributed by the 3-block $\intmodel{B}$ is now
\begin{equation*}
C =  4x^2   + x(2r_2 + 2r_4 + 2r_5 + b_1+2b_2+b_3+b_4+b_5)
-y_2(y_2 + \Delta_2 + \Delta_4 + \Delta_5).
\end{equation*} 

Notice that $C-C_{RBR} = -y_2(y_2 + \Delta_2 + \Delta_4 + \Delta_5)$.
Thus, $C\geq C_{RBR}$ only if $0\leq y_2 \leq -(\Delta_2+\Delta_4+\Delta_5)$.
As $|\Delta_2+\Delta_4+\Delta_5|\leq\eta$, we are done.
The case when $y_1=y_3=x$ is treated similarly.

\end{proof}

\begin{lemma}\label{lemma:edge_gadget_stable}
Let $\intmodel{E}=\block^1\sqcup\block^2\sqcup\{S_{12}\}$ be an edge gadget of $H$ together with two long intervals $L_1,L_2$ terminating at $\block_1^1$ and $\block_2^2$.
Suppose that $\intmodel{E}$ is covered by at most $\eta$ long intervals.
Then, 
\begin{inparaenum}[(i)]
\item the colorings of 3-blocks are almost alternating except for at most $\eta$ intervals of the central blocks, and 
\item if $\col(L_1)\not=\col(L_2)$, then the $\maxcut$ value provided by $\intmodel{E}$ is greater than the $\maxcut$ value in the case $\col(L_1)=\col(L_2)$, by at least $k-2\eta$.
\end{inparaenum}
\end{lemma}
\begin{proof}
Suppose that $\intmodel{E}$ is covered by $\lambda<\eta$ long intervals. By \Cref{lemma:3block_stable}, the $\maxcut$ of $\block^1$ and $\block^2$, 
is either alternating or almost alternating except for $\Delta$ intervals of $\block_2^1$ or $\block_2^2$, where $\Delta<\eta$ is the difference between the numbers of Red and Blue long covering intervals.

Suppose that $\col(L_1)=\col(L_2)=B$; then $\col(\block_1^1)=\col(\block_3^1)=R$. The color of $\block_2^1$ must be Blue except for at most $\eta$ intervals. For $\block^2$ there are two cases. 
\begin{enumerate}
    \item The color of $\block_2^2$ is Red except for at most $\eta$ intervals, and $\col(\block_1^2)=\col(\block_3^2)=B$. Then, for any $\col(S_{12})$, the cut size is at most  $8k^2+(6+4\lambda)k + \frac{1}{4}\Delta^2+\frac{3}{2}\lambda+1$.
    \item The color of $\block_2^2$ is Blue except for at most $\eta$ intervals, and $\col(\block_1^2)=\col(\block_3^2)=R$. Then $\col(S_{12})=B$. Then the cut size is at most $8k^2+(6+4\lambda)k+\frac{1}{2}\Delta^2+2\lambda$.
\end{enumerate}

Suppose that $\col(L_1)=B$ and $\col(L_2)=R$. Then $\block^2$ is partitioned as $RBR$. $\block^1$ is also partitioned as $RBR$. Then the cut size would be at least $8k^2+(7+4\lambda)k+\frac{1}{2}\Delta^2+2$.

The difference between the value provided by $\intmodel{E}$ when $\col(L_1)\not=\col(L_2)$ and the value in the case when they are equal is at least $k-2\lambda$, which is at least $k-2\eta$.
\end{proof}

\begin{lemma}\label{lemma:switch_gadget_stable}
Suppose that the first part of a switch gadget is covered by $\eta$ long intervals. 
Then, in any its $\maxcut$ partition, the coloring of $\block_1^{\text{bot}},\ldots,\block_9^{\text{bot}}$ is alternating and the coloring of $\block_1^{\text{top}},\ldots,\block_4^{\text{top}}$ is alternating except for at most $\eta$ intervals within $\block_1^{\text{top}}$ or $\block_4^{\text{top}}$.
\end{lemma}
\begin{proof}
Let the number of Red and Blue long intervals covering the switch gadget be $r$ and $b$ respectively, and denote the quantity $(r-b)$ by $\Delta$. 
Suppose that $r>b$. Let $y_i,z_j$, for $i\in\{1,\ldots,4\},j\in\{1,\ldots,9\}$, be the same as in \Cref{lemma:switch_gadget}. Compute the size of the cut:
\begin{multline*}
    \cut = 12x'^2 + 28x^2 + 16x'x \\
    -(y_1-y_2-z_4+z_5)^2-(y_3-y_4+z_5-z_6)^2\\
    - 2\left(z_1-\frac{z_2}{2}\right)^2 -(z_2-z_3)^2 -(z_3-z_4)^2 -(z_6-z_7)^2-(z_7-z_8)^2- 2\left(\frac{z_8}{2}-z_9\right)^2\\
    -(y_2-y_3)^2 -\sum_{i=2}^3(2x'-y_i)y_i\\
    -(x-z_1)z_1-\left(x-\frac{z_2}{2}\right)z_2-\sum_{j=3}^7(2x-z_j)z_j - \left(x-\frac{z_8}{2}\right)z_8-(x-z_9)z_9\\
    -\Delta\left(\sum_{i=1}^4(-1)^iy_i+\sum_{j=1}^9(-1)^jz_j\right) + (8x+4x')(r+b)
\end{multline*}
\begin{multline*}
    = 12x'^2 + 28x^2 + 16x'x \\
    -(y_1-y_2-z_4+z_5-\Delta/4)^2-(y_3-y_4+z_5-z_6-\Delta/4)^2\\
    - 2\left(z_1-\frac{z_2}{2}-\Delta/4\right)^2 -(z_2-z_3+\Delta/4)^2 -(z_3-z_4-\Delta/4)^2 \\
    -(z_6-z_7+\Delta/4)^2-(z_7-z_8-\Delta/4)^2- 2\left(\frac{z_8}{2}-z_9+\Delta/4\right)^2\\
    -(y_2-y_3+\Delta/4)^2 -\sum_{i=2}^3(2x'-y_i)y_i\\
    -(x-z_1)z_1-\left(x-\frac{z_2}{2}\right)z_2-\sum_{j=3}^7(2x-z_j)z_j - \left(x-\frac{z_8}{2}\right)z_8-(x-z_9)z_9\\
    + \frac{\Delta}{2}(y_1-y_4)+ (8x+4x')(r+b) +\frac{11}{16}\Delta^2.
\end{multline*}
Consider only those summands that  participate in the size of the cut when the coloring of the blocks alternates, i.e., when, for all $i,j$, $y_i=z_j=0$:
\begin{equation*}
\cut_{\text{alter}} = 12x'^2+28x^2+16x'x+(8x+4x')(r+b).    
\end{equation*}
If we choose $x>\frac{(2\Delta+1)^2}{2}+\frac{11}{4}\Delta^2$, then the distance between some variable (except for $y_1,y_4$) and the closest end of its domain cannot be greater than $\Delta$. For example, suppose that $\min(z_1,x-z_1)>\Delta$. Then $(x-z_1)z_1>\Delta x+\frac{11}{16}\Delta^2$, and so
\begin{equation*}
    \cut \leq \cut_{\text{alter}} - (x-z_1)z_1 + \frac{\Delta}{2}(y_1-y_4) + \frac{11}{16}\Delta^2 < \cut_{\text{alter}}.
\end{equation*}
So, in this case it will not be a maximum cut. 

Suppose that $|z_j-z_{j+1}|>\Delta$, then either $|z_j-z_{j+1}|>2x-2\Delta$ when $j\not\in\{1,4,5,8\}$, or $|z_j-z_{j+1}|>x-2\Delta$ when $j\in\{1,8\}$. But then one of the clauses will be greater than $(x-2\Delta-\Delta/4)^2$, hence, greater than $\Delta x+\frac{11}{16}\Delta^2$. Similarly, $|y_2-y_3|<\Delta$. Consider the variables $z_4,z_5,z_6$. Choose $x'$ between $x/2$ and $x$ such that $(2x'-2x+2\Delta+\Delta/4)^2$ is greater than $\Delta x+\frac{11}{16}\Delta^2$, as $x>\Delta^2$, it is easy to choose a convenient value for $x'$. For such $x'$, we will have $|z_4-z_5|,|z_5-z_6|<\Delta$.

Denote $d_{12}=y_1-y_2$ and $d_{34}=y_3-y_4$. Then $y_1-y_4 < |d_{12}| + \Delta+|d_{34}|$. Then we can choose only those $d_{12},d_{34}$ that satisfy
\begin{equation*}
    -(d_{12}-z_4+z_5-\Delta/4)^2-(d_{34}+z_5-z_6-\Delta/4)^2+\frac{\Delta}{2}(|d_{12}|+\Delta+|d_{34}|)+\frac{11}{16}\Delta^2 >0
\end{equation*}
as, otherwise, it will not be a maximal cut. Denote $w = -z_4+z_5-\Delta/4, w'=z_5-z_6-\Delta_4$, we know that $|w|,|w'|\leq5\Delta/4$. Then the expression can be written in the following form:
\begin{equation*}
    -d_{12}^2 + 2d_{12}w-w^2-d_{34}^2+2d_{34}w'-w'^2 + \Delta/2(|d_{12}|+|d_{34}|)+\frac{19}{16}\Delta^2>0
\end{equation*}
\begin{equation*}
    \Rightarrow -(d_{12} - w-\Delta/4)^2-(d_{34}-w'-\Delta/4)^2+\frac{\Delta}{2}(|w|+|w'|) + \frac{2\Delta^2}{16} + \frac{19}{16}\Delta^2>0.
\end{equation*}

Take $|d_{12}| \geq 4\Delta$. Then 
\begin{equation*}
    -(d_{12} - w-\Delta/4)^2-(d_{34}-w'-\Delta/4)^2+\frac{\Delta}{2}(|w|+|w'|) + \frac{2\Delta^2}{16} + \frac{19}{16}\Delta^2 
\end{equation*}
\begin{equation*}
    \leq -\left(\frac{5\Delta}{2}\right)^2 + \frac{5\Delta^2}{4} + \frac{21\Delta^2}{16}<0.
\end{equation*}
The same is true for $|d_{34}|$. We now can imply that $|y_1-y_4|<9\Delta$, otherwise the cut is not maximal. This is an important result because we can show now that all the blocks except for $y_1,y_4$ are monochromatic. Take any $x>\frac{9}{2}\Delta^2+\frac{11}{16}\Delta^2 = \frac{83}{16}\Delta^2$, then suppose that for some variable except for $y_1,y_4$, say $z_1$ for example: $z_1\not\in\{0,x\}$. Then $(x-z_1)z_1\geq x-1 >\frac{9}{2}\Delta^2+\frac{11}{16}\Delta^2$. So the cut will not be maximal in this case. Same is true for any variable other than $y_1,y_4$. So we can conclude that $2z_1=z_2=\ldots=z_8=2z_9\in\{0,2x\}$, and $y_2=y_3\in\{0,2x'\}$.

Suppose w.l.o.g. that all the variables except for $y_1,y_4$ are equal to $0$. Then $\cut-\cut_{\text{alter}}$ equals
\begin{equation*}
     -(y_1-\Delta/4)^2 - (y_4+\Delta/4)^2 +\frac{\Delta}{2}(y_1-y_4)+ \frac{\Delta^2}{8} = -y_1(y_1 -\Delta) -y_4(y_4 + \Delta).
\end{equation*}
One can see now that the cut will be optimal if $y_1=\Delta/2,y_4=0$ for $\Delta>0$. For the case when $\Delta<0$ the optimal cut will be when $y_1=0$ and $y_4=-\Delta/2$. 
Recall that $\eta \geq r+b$. As $\Delta = (r-b)$, we have $|\Delta/2| < \eta$.
\end{proof}

\paragraph*{Reduction}

At first, we return to the conditions about $H$ that we stated at the beginning of \Cref{section:gadgets} and that we assumed to hold. 

\begin{enumerate}
\item Every block of $H$ is covered by at most $\eta := 3n$ long intervals. 
\item Every long interval of $H$ covers at most $\mu:=30n+16$ blocks.
\end{enumerate}

We need to verify that the graph $H$ that we have just constructed indeed satisfies them. 
In total, there are $3n$ long intervals starting from vertex gadgets, so every block is intersected by at most $3n$ long intervals.  
Every long interval covers $n-1$ zones and $n$ buffers, each zone contains a join gadget that has at most ten 3-blocks, each buffer contains at most one switch gadget that has 16 blocks. 
Thus, those 3 conditions are satisfied for $H$.


Let us say that a partition of $V(H)$ in 2 classes $\{R,B\}$ is \emph{good} if the following conditions are satisfied.
\begin{enumerate}
    \item For any gadget, its coloring is alternating or almost alternating except for $\block_2^i$ of a 3-block $\block^i$ or one of $\block_1^{\text{top}},\block_4^{\text{top}}$ of a switch gadget that contain at most $\eta$ intervals of the other color.
    \item For any short or long interval $L$ that starts from or terminates at a block $\block$ of size $x',x$, or $2x$, we have $\col(L)\not=\col(\block)$.
\end{enumerate}

\begin{lemma}\label{lemma:reduction_good}
Any $\maxcut$ partition of $H$ is good.
\end{lemma}
\begin{proof}
The condition (1) is provided by \Cref{lemma:switch_gadget_stable,lemma:3block_stable}. 
Let $L$ be the leftmost interval that does not satisfy (2), suppose that it starts from a block $\block$ of size at least $x'>\frac{x}{2}$ and $\col(L)=\col(\block)$. $L$ is a part of some link chain that connects a vertex gadget $\intmodel{V}$ to an edge gadget $\intmodel{E}$. Invert $\col(L)$ and modify the colors of all gadgets and short link and long intervals of this chain that are between $L$ and $\intmodel{E}$ so that all intervals of this chain satisfy (2). The gain after this operation is at least $\frac{x}{2}$ because $\col(L)\not=\col(\block)$. Denote by $l$ the number of short link and long intervals in the chain, it is at most $10\times3n\times 9\times \frac{3n}{2}$. The loss is at most the sum of the following values:
\begin{compactitem}
\item $k$, if both intervals terminating at $\intmodel{E}$ now have the same color;
\item $2\eta\mu l$ -- the cut between the gadgets of $H$ covered by the intervals of the chain, where $\eta$ is an upper bound for the number of long intervals covering a gadget and $\mu$ is an upper bound on blocks that a long interval intersects;
\item $2\eta l$ -- the cut between the long intervals of $H$ intersected by the intervals of the chain;
\item $2\eta^2 l$ -- the cut between the gadgets of the chain and the long intervals of $H$ covering them (the number of gadgets in the chain also equals $l$). 
\end{compactitem}
As $\eta,\mu\in O(n)$, and $l\in O(n^2)$, the loss is in $O(k+n^4)$. We are free to choose $x$ to be large enough for the gain to be strictly greater than the maximal possible loss. Thus, for such value of $x$, the second condition also holds.
\end{proof}

A good partition $p\in \{R, B\}^{V(H)}$ corresponds to some partition $q$ of $V(G)$: assign to a vertex $g_i\in V(G)$ the same color that is assigned to the long intervals starting from the vertex gadget $\intmodel{V}_i$. For such $p$ and $q$, we will write $p\sim q$. Clearly, the other direction is also true: for every $q\in \{R, B\}^{V(G)}$ there exists a good partition $p$ such that $p\sim q$. Recall that each gadget is covered by at most $\eta$ long intervals, there are $\eta$ long interval chains, and each long interval covers at most $\mu$ gadgets. The following Lemma~\ref{lemma:reduction_iff} implies that $\maxcut$ on cubic graphs reduces to $\maxcut$ on interval graphs of interval count 2. It is well-known that a $\maxcut$ problem is in NP, so, as the size of $H$ is polynomial in $|G|$, it implies \Cref{th:main_result}.

\begin{lemma}\label{lemma:reduction_iff}
Let $p\in \{R, B\}^{V(H)}$ be a $\maxcut$ partition of $H$, and $q\in \{R, B\}^{V(G)}$ be a partition of $G$ such that $p\sim q$. Then $q$ is a $\maxcut$ partition of $G$.
\end{lemma}
\begin{proof}
Suppose that $q$ is not a $\maxcut$ partition, that is, there is another $q'\in \{R, B\}^{V(G)}$ which is maximal. Let $p'\in \{R, B\}^{V(H)}$ be a good partition that satisfies $p'\sim q'$.

For simplicity, we denote by $E$ the set of intervals that belong to edge gadgets, and  $C$ denotes the set of intervals that belong to chains or to other gadgets. Clearly, $V(H) = C\sqcup E$.

The difference between the cut values of $p$ and $p'$, for edges induced by $E$, is at most $\eta^3\in O(n^3)$ because there are $\frac{\eta}{2}$ edge gadgets, and each of them changes the cut by at most $2\eta^2$. The difference between the cut values, for edges induced by $C$, is at most 
\begin{equation*}
\underbrace{\eta^2\cdot 10\cdot l\cdot \eta}_{\text{within each gadget}}+\underbrace{2\cdot\eta\cdot l\cdot\eta}_{\text{between long intervals}} + \underbrace{10\cdot\eta\cdot\mu\cdot l\cdot\eta}_{\text{between gadgets and long intervals}}
\end{equation*}
which belongs to $O(n^5)$. Finally, as $q'$ has a strictly greater cut value than $q$, we know that the number of cut edges between $C$ and $E$ for $p'$ is greater than the corresponding number for $p$ by at least $(k-2\eta) - \eta\cdot2\eta\cdot \frac{\eta}{2}$. By \Cref{lemma:edge_gadget_stable} and because an edge gadget is covered by at most $\eta$ long intervals, each of them adds at most $2\eta$ to the cut, and there are $\frac{\eta}{2}$ edge gadgets.
As we choose $k$ to be in $\Omega(n^6)$, the cut value of $p'$ is greater than the cut value of $p$, it is a contradiction.
\end{proof}

\section{Conclusion}

In this article, we have shown that $\maxcut$ remains NP-complete for interval graphs of interval count two.
The future work in this direction is to understand the complexity of $\maxcut$ on unit interval graphs.
Apart from that, our result raises the following question.
\begin{question}
Let $C\subset\mathbb{R}$ be a fixed set of size $\ell$.
What is the complexity of Maximum Cut on interval graphs that have a model, where the length of each interval is in $C$?
\end{question}
In the current paper, the length of long intervals depends on the number of vertices of the cubic graph $G$, therefore, it may be arbitrarily large.
This question asks, whether the problem remains NP-complete if there is an additional restriction on how much different are the lengths of short and long intervals.
If the length of short intervals is 1, and the length of long intervals is $c$, then, by making the value of $c$ smaller, the graph class becomes more similar to unit interval graphs that correspond to the case when $c=1$.

\section*{Acknowledgements}
The authors are thankful to Kaustav Bose for correcting \Cref{lemma:3block_stable}.

\section*{Declaration}
Alexey Barsukov is funded by the European Union (ERC, POCOCOP, 101071674). Views and opinions expressed are however those of the author(s) only and do not necessarily reflect those of the European Union or the European Research Council Executive Agency. Neither the European Union nor the granting authority can be held responsible for them.

Bodhayan Roy is funded by SERB MATRICS, Grant Number MTR/2021/000474.

\bibliography{maxcut2}

\begin{thebibliography}{10}
\expandafter\ifx\csname url\endcsname\relax
  \def\url#1{\burl{#1}}\fi
\expandafter\ifx\csname urlprefix\endcsname\relax\def\urlprefix{URL }\fi
\providecommand{\bibinfo}[2]{#2}
\providecommand{\eprint}[2][]{\url{#2}}
\providecommand{\doi}[1]{\url{https://doi.org/#1}}
\bibcommenthead

\bibitem{interval}
\bibinfo{author}{Adhikary, R.}, \bibinfo{author}{Bose, K.},
  \bibinfo{author}{Mukherjee, S.} \& \bibinfo{author}{Roy, B.}
\newblock \bibinfo{title}{Complexity of maximum cut on interval graphs}.
\newblock \emph{\bibinfo{journal}{Discrete \& Computational Geometry}}
  \textbf{\bibinfo{volume}{70}}, \bibinfo{pages}{307--322}
  (\bibinfo{year}{2023}).
\newblock \urlprefix\url{https://doi.org/10.1007/s00454-022-00472-y}.

\bibitem{interval4}
\bibinfo{author}{de~Figueiredo, C. M.~H.}, \bibinfo{author}{de~Melo, A.~A.},
  \bibinfo{author}{Oliveira, F.~S.} \& \bibinfo{author}{Silva, A.}
\newblock \bibinfo{title}{Maximum cut on interval graphs of interval count four
  is np-complete}.
\newblock \emph{\bibinfo{journal}{Discrete \& Computational Geometry}}
  (\bibinfo{year}{2023}).
\newblock \urlprefix\url{https://doi.org/10.1007/s00454-023-00508-x}.

\bibitem{garey_johnson1990}
\bibinfo{author}{Garey, M.~R.} \& \bibinfo{author}{Johnson, D.~S.}
\newblock \emph{\bibinfo{title}{Computers and Intractability: A Guide to the
  Theory of NP-Completeness}}  (\bibinfo{publisher}{W. H. Freeman \& Co.},
  \bibinfo{address}{USA}, \bibinfo{year}{1979}).

\bibitem{karp1972}
\bibinfo{author}{Karp, R.~M.}
\newblock \emph{\bibinfo{title}{Reducibility among Combinatorial Problems}},
  \bibinfo{pages}{85--103} (\bibinfo{publisher}{Springer US},
  \bibinfo{address}{Boston, MA}, \bibinfo{year}{1972}).
\newblock \urlprefix\url{https://doi.org/10.1007/978-1-4684-2001-2_9}.

\bibitem{cubic}
\bibinfo{author}{Berman, P.} \& \bibinfo{author}{Karpinski, M.}
\newblock \bibinfo{editor}{Wiedermann, J.}, \bibinfo{editor}{van Emde~Boas, P.}
  \& \bibinfo{editor}{Nielsen, M.} (eds) \emph{\bibinfo{title}{On some tighter
  inapproximability results (extended abstract)}}.
\newblock (eds \bibinfo{editor}{Wiedermann, J.}, \bibinfo{editor}{van
  Emde~Boas, P.} \& \bibinfo{editor}{Nielsen, M.})
  \emph{\bibinfo{booktitle}{Automata, Languages and Programming}},
  \bibinfo{pages}{200--209} (\bibinfo{publisher}{Springer Berlin Heidelberg},
  \bibinfo{address}{Berlin, Heidelberg}, \bibinfo{year}{1999}).

\bibitem{split_cobipartite_btw}
\bibinfo{author}{Bodlaender, H.~L.} \& \bibinfo{author}{Jansen, K.}
\newblock \bibinfo{title}{On the complexity of the maximum cut problem}.
\newblock \emph{\bibinfo{journal}{Nordic J. of Computing}}
  \textbf{\bibinfo{volume}{7}}, \bibinfo{pages}{14–31}
  (\bibinfo{year}{2000}).

\bibitem{unit_disk}
\bibinfo{author}{D{\'\i}az, J.} \& \bibinfo{author}{Kami{\'n}ski, M.}
\newblock \bibinfo{title}{Max-cut and max-bisection are np-hard on unit disk
  graphs}.
\newblock \emph{\bibinfo{journal}{Theoretical Computer Science}}
  \textbf{\bibinfo{volume}{377}}, \bibinfo{pages}{271--276}
  (\bibinfo{year}{2007}).
\newblock
  \urlprefix\url{https://www.sciencedirect.com/science/article/pii/S0304397507000801}.

\bibitem{total_line}
\bibinfo{author}{Guruswami, V.}
\newblock \bibinfo{title}{Maximum cut on line and total graphs}.
\newblock \emph{\bibinfo{journal}{Discrete Applied Mathematics}}
  \textbf{\bibinfo{volume}{92}}, \bibinfo{pages}{217--221}
  (\bibinfo{year}{1999}).
\newblock
  \urlprefix\url{https://www.sciencedirect.com/science/article/pii/S0166218X99000566}.

\bibitem{planar}
\bibinfo{author}{Hadlock, F.}
\newblock \bibinfo{title}{Finding a maximum cut of a planar graph in polynomial
  time}.
\newblock \emph{\bibinfo{journal}{SIAM Journal on Computing}}
  \textbf{\bibinfo{volume}{4}}, \bibinfo{pages}{221--225}
  (\bibinfo{year}{1975}).
\newblock \urlprefix\url{https://doi.org/10.1137/0204019}.

\bibitem{not_contractible_to_k5}
\bibinfo{author}{Barahona, F.}
\newblock \bibinfo{title}{The max-cut problem on graphs not contractible to
  k5}.
\newblock \emph{\bibinfo{journal}{Oper. Res. Lett.}}
  \textbf{\bibinfo{volume}{2}}, \bibinfo{pages}{107–111}
  (\bibinfo{year}{1983}).
\newblock \urlprefix\url{https://doi.org/10.1016/0167-6377(83)90016-0}.

\bibitem{permutation}
\bibinfo{author}{de~Figueiredo, C. M.~H.}, \bibinfo{author}{de~Melo, A.~A.},
  \bibinfo{author}{Oliveira, F.~S.} \& \bibinfo{author}{Silva, A.}
\newblock \bibinfo{title}{Maxcut on permutation graphs is np-complete}.
\newblock \emph{\bibinfo{journal}{Journal of Graph Theory}}
  \textbf{\bibinfo{volume}{104}}, \bibinfo{pages}{5--16}
  (\bibinfo{year}{2023}).
\newblock
  \urlprefix\url{https://onlinelibrary.wiley.com/doi/abs/10.1002/jgt.22948}.

\bibitem{unit_interval1}
\bibinfo{author}{Bodlaender, H.~L.}, \bibinfo{author}{Kloks, T.} \&
  \bibinfo{author}{Niedermeier, R.}
\newblock \bibinfo{title}{Simple max-cut for unit interval graphs and graphs
  with few p4s}.
\newblock \emph{\bibinfo{journal}{Electronic Notes in Discrete Mathematics}}
  \textbf{\bibinfo{volume}{3}}, \bibinfo{pages}{19--26} (\bibinfo{year}{1999}).
\newblock
  \urlprefix\url{https://www.sciencedirect.com/science/article/pii/S1571065305800149}.
\newblock \bibinfo{note}{6th Twente Workshop on Graphs and Combinatorial
  Optimization}.

\bibitem{unit_interval2}
\bibinfo{author}{Boyacı, A.}, \bibinfo{author}{Ekim, T.} \&
  \bibinfo{author}{Shalom, M.}
\newblock \bibinfo{title}{A polynomial-time algorithm for the maximum
  cardinality cut problem in proper interval graphs}.
\newblock \emph{\bibinfo{journal}{Information Processing Letters}}
  \textbf{\bibinfo{volume}{121}}, \bibinfo{pages}{29--33}
  (\bibinfo{year}{2017}).
\newblock
  \urlprefix\url{https://www.sciencedirect.com/science/article/pii/S0020019017300133}.

\bibitem{unit_interval1correction}
\bibinfo{author}{Bodlaender, H.~L.}, \bibinfo{author}{de~Figueiredo, C. M.~H.},
  \bibinfo{author}{Gutierrez, M.}, \bibinfo{author}{Kloks, T.} \&
  \bibinfo{author}{Niedermeier, R.}
\newblock \bibinfo{editor}{Ribeiro, C.~C.} \& \bibinfo{editor}{Martins, S.~L.}
  (eds) \emph{\bibinfo{title}{Simple max-cut for split-indifference graphs and
  graphs with few p4's}}.
\newblock (eds \bibinfo{editor}{Ribeiro, C.~C.} \& \bibinfo{editor}{Martins,
  S.~L.}) \emph{\bibinfo{booktitle}{Experimental and Efficient Algorithms}},
  \bibinfo{pages}{87--99} (\bibinfo{publisher}{Springer Berlin Heidelberg},
  \bibinfo{address}{Berlin, Heidelberg}, \bibinfo{year}{2004}).

\bibitem{unit_interval2correction}
\bibinfo{author}{Kratochv{\'\i}l, J.}, \bibinfo{author}{Masa{\v r}{\'\i}k, T.}
  \& \bibinfo{author}{Novotn{\'a}, J.}
\newblock \bibinfo{title}{{\$}{\$}{\{}{$\backslash$}mathcal
  {\{}u{\}}{\}}{\$}{\$}-bubble model for mixed unit interval graphs and its
  applications: The maxcut problem revisited}.
\newblock \emph{\bibinfo{journal}{Algorithmica}} \textbf{\bibinfo{volume}{83}},
  \bibinfo{pages}{3649--3680} (\bibinfo{year}{2021}).
\newblock \urlprefix\url{https://doi.org/10.1007/s00453-021-00837-4}.

\end{thebibliography}


\end{document}